\documentclass{article}
\usepackage{amsfonts}
\usepackage{amsmath}

\newtheorem{theorem}{Theorem}
\newtheorem{acknowledgement}[theorem]{Acknowledgement}

\newtheorem{definition}[theorem]{Definition}

\newtheorem{lemma}[theorem]{Lemma}

\newenvironment{proof}[1][Proof]{\noindent\textbf{#1.} }{\ \rule{0.5em}{0.5em}}
\input{tcilatex}
\begin{document}

\title{Derivatives of Complete Weight Enumerators and New Balance Principle
of Binary Self-Dual Codes}
\author{Vassil Yorgov \\
Department of Mathematics and Computer Science\\
Fayetteville State University\\
1200 Murchison Rd\\
Fayetteville, NC 28311}
\maketitle

\begin{abstract}
Let $K=\frac{1}{\sqrt{2}}%
\begin{bmatrix}
1 & 1 \\ 
1 & -1%
\end{bmatrix}%
.$ It is known that the complete weight enumerator $\ W$ of a binary
self-dual code of length $n$ is an eigenvector corresponding to an
eigenvalue 1 of the Kronecker power $K^{[n]}.$ For every integer $t,~0\leq
t\leq n,$ we define the t-th derivative $W_{<t>}$ of $W$ in such a way that $%
W_{<t>}$ is in the eigenspace of $\ 1$ of the matrix $K^{[n-t]}.$ For large
values of $t,$ $W_{<t>}$ contains less information about the code but has
smaller length while $W_{<0>}=W$ completely determines the code. We compute
the derivative of order $n-5$ for the extended Golay code of length 24, the
extended quadratic residue code of length 48, and the putative [72,24,12]
code and show that they are in the eigenspace of $\ 1$ of the matrix $%
K^{[5]}.$ We use the derivatives to prove a new balance equation which
involves the number of code vectors of given weight having 1 in a selected
coordinate position. As an example, we use the balance equation to eliminate
some candidates for weight enumerators of binary self-dual codes of length
eight.
\end{abstract}

\section{Introduction}

In this work, we use the standard notation of error correcting codes. An
appropriate reference is the book \cite{Huff-Pless}. Let $C$ be a binary
linear self-dual code of length $n.$ Thus $C$ is a linear subspace of
dimension $\ n/2$ of the vector space $F^{n},$ where $F$ is the binary
field, and $C$ is equal to its orthogonal code with respect to the usual dot
product. The exact weight enumerator of $C$ is a row vector $W$ of length $%
2^{n}$ with integer entries labeled with the vectors from $F^{n}$ listed in
lexicographic order. The entry $W[v]$ labeled with a vector $v$ is defined
with

\begin{equation*}
W[v]=\left\{ 
\begin{tabular}{cc}
$1$ & if $v\in C$ \\ 
$0$ & if $v\notin C.$%
\end{tabular}%
\right.
\end{equation*}%
Thus $W$ belongs to the vector space $Q^{2^{n}}$ where $Q$ is the rational
field. Clearly the exact weight enumerator $W$ determines completely the
code $C.$

Let $K=\frac{1}{\sqrt{2}}%
\begin{bmatrix}
1 & 1 \\ 
1 & -1%
\end{bmatrix}%
,$ $K^{[2]}=K\otimes K,$ $K^{[3]}=K^{[2]}\otimes K,$ and let $%
K^{[n]}=K^{[n-1]}\otimes K,$ be the $n$-th Kronecker power of $K.$ A\
summery of the properties of Kronecker products of matrices is given in \cite%
{Zhang}. Let $R$ be the real number field. The matrix $K$ has eigenvalues $1$
and $-1.$ The vectors $\left( 1,\sqrt{2}-1\right) $ and $\left( 1,-\sqrt{2}%
-1\right) $ are eigenvectors of $K$ corresponding to $1$ and $-1.$Therefore, 
$K$ is simple, meaning that there exist a basis of $R^{2}$ consisting of
eigenvectors of $K.$ The rows of the matrix \ 
\begin{equation*}
B=%
\begin{bmatrix}
1 & \sqrt{2}-1 \\ 
1 & -\sqrt{2}-1%
\end{bmatrix}%
\end{equation*}%
form such basis of $R^{2}.$ It follows that the matrix $K^{[n]}$ is also
simple and has eigenvalues $1$ and $-1$ , \cite{Lan} page 413. The rows of
the Kronecker power $B^{[n]}$ form a basis of $R^{2^{n}}.$ We label the rows
of $B^{[n]}$ with the vectors from $F^{n}$ listed in lexicographic order.

\begin{lemma}
The rows of $B^{[n]}$ with even weight labels form a basis of the eigenspace
of $1$ of $K^{[n]}.$
\end{lemma}

\begin{proof}
We have $BK=DB$ where 
\begin{equation*}
D=%
\begin{bmatrix}
1 & 0 \\ 
0 & -1%
\end{bmatrix}%
.
\end{equation*}%
Using the properties of Kronecker power, we obtain $\left( BK\right)
^{[n]}=\left( DB\right) ^{[n]}$ and $B^{[n]}K^{[n]}=D^{[n]}B^{[n]}.$ Since
the diagonal matrix $D^{[n]}$ has $1$ on the rows labeled with even weight
vectors and $-1$ on the rows labeled with odd weight vectors, the lemma
follows from the last equality.
\end{proof}

The next lemma is a particular case of a result proved in \cite{Sim}.

\begin{lemma}
\cite{Sim}\label{Lsim} The complete weight enumerator $W$ of the binary
self-dual code $C$ belongs to the eigenspace of $1$ of the matrix $K^{[n]}$.
Moreover, any nontrivial (0,1) vector from the eigenspace of $1$ of the
matrix $K^{[n]}$ determines a self-dual binary code.
\end{lemma}

Since $2^{n}$ increases very rapidly with $n,$ this lemma is not practical
when the code length $n$ increases. In the next section we define
derivatives of $W$ which contain less information about the code $C$ but
have smaller size.

\section{Derivatives of the Complete Weight Enumerator}

Let $\rho =\sqrt{2}-1$ and $\mu =-\sqrt{2}-1.$ The quadratic extension $%
Q(\rho )$ has automorphism group of order two generated by the automorphism $%
\Phi :Q(\rho )\rightarrow Q(\rho )$ defined with $\Phi (\rho )=\mu .$

\begin{definition}
\label{defDer}For every integer $t,$ $0\leq t\leq n,$ the $t$-th derivative
of $W$ is the row vector $W_{<t>}$ of length $2^{n-t}$ having an entry with
label $v\in F^{n-t}$ determined by 
\begin{equation*}
W_{<t>}[v]=\sum_{u\in F^{t}}\rho ^{wt(u)}W[uv]
\end{equation*}%
where $wt(u)$ is the Hemming weight of $u$ and $uv$ is the concatenation of $%
u$ and $v.$
\end{definition}

For example, if $u=(1,1,0,1)$ and $v=(1,0),$ then $uv=(1,1,0,1,1,0).$ It
follows that every derivative $W_{<t>}$ has nonnegative entries from the
quadratic extension $Q(\rho ).$ Particularly, $W_{<n>}$ is determined by the
weight distribution of $C.$

\begin{lemma}
\begin{equation*}
W_{<n>}=\sum_{u\in F^{n}}\rho ^{wt(u)}W[u]=\sum_{k=0}^{n}A_{k}\rho ^{k}
\end{equation*}
\end{lemma}

where $A_{k}$ is the number of weight $k$ vectors in $C.$

Thus $W_{<n>}\in Q(\rho )$ and contains less information about $C$ than the
weight distribution of $C.$ On the other end, $W_{<0>}=W$ determines the
code $C$ completely.

\begin{lemma}
\label{nextD}For every integer $t,$ $0\leq t\leq n-1$ and for every $v\in
F^{n-t-1},$ we have $W_{<t+1>}[v]=W_{<t>}[(0)v]+\rho W_{<t>}[(1)v].$
\end{lemma}

\begin{proof}
The proof is straightforward: 
\begin{eqnarray*}
W_{<t+1>}[v] &=&\sum_{u^{\prime }\in F^{t+1}}\rho ^{wt(u^{\prime
})}W[u^{\prime }v] \\
&=&\sum_{u\in F^{t}}\rho ^{wt(u(0))}W[u(0)v]+\sum_{u\in F^{t}}\rho
^{wt(u(1))}W[u(1)v] \\
&=&\sum_{u\in F^{t}}\rho ^{wt(u)}W[u(0)v]+\rho \sum_{u\in F^{t}}\rho
^{wt(u)}W[u(1)v] \\
&=&W_{<t>}[(0)v]+\rho W_{<t>}[(1)v].
\end{eqnarray*}
\end{proof}

Less formally, $W_{<t+1>}=$ $W_{<t>}[0]+\rho W_{<t>}[1]$ where $W_{<t>}[0]$
is the first half and $W_{<t>}[1]$ is the second half of $W_{<t>}$.

The next lemma shows that every entry of the second half of $W_{<t>}$
depends only on a corresponding entry from the first half of $W_{<t>}.$ 

\begin{lemma}
\label{2half}For every integer $t,$ $0\leq t\leq n-1$ and for every $v\in
F^{n-t-1},$ we have $W_{<t>}[(1)\overline{v}]=(-1)^{wt(v)}\rho ^{t}\Phi
(W_{<t>}[(0)v])$ where $\overline{v}$ is the complementary vector of $v$ and 
$\Phi $ is the automorphism of $Q(\rho )$ defined above. 
\end{lemma}

\begin{proof}
As the all one vector is in the code $C,$ for any $u\in F^{t}$ we have $W[%
\overline{u}(1)\overline{v}]=W[u(0)v]$ where $\overline{u}$ and $\overline{v}
$ are the complementary vectors of $u$ and $v.$  By Definition \ref{defDer}, 
\begin{eqnarray*}
W_{<t>}[(1)\overline{v}] &=&\sum_{\overline{u}\in F^{t}}\rho ^{wt(\overline{u%
})}W[\overline{u}(1)\overline{v}] \\
&=&\sum_{\overline{u}\in F^{t}}\rho ^{wt(\overline{u})}W[u(0)v] \\
&=&\sum_{u\in F^{t}}\rho ^{t-wt(u)}W[u(0)v] \\
&=&\rho ^{t}\sum_{u\in F^{t}}\rho ^{-wt(u)}W[u(0)v]
\end{eqnarray*}%
We check that $\rho ^{-1}=-\mu .$ Because $C$ contains only even weight
vectors, $W[u(0)v]=0$ when $wt(u)$ and $wt(v)$ have different parities.
Therefore%
\begin{eqnarray*}
W_{<t>}[(1)\overline{v}] &=&\rho ^{t}\sum_{u\in F^{t}}(-\mu )^{wt(u)}W[u(0)v]
\\
&=&\rho ^{t}(-1)^{wt(v)}\sum_{u\in F^{t}}\mu ^{wt(u)}W[u(0)v] \\
&=&\rho ^{t}(-1)^{wt(v)}\Phi \left( \sum_{u\in F^{t}}\rho
^{wt(u)}W[u(0)v]\right)  \\
&=&(-1)^{wt(v)}\rho ^{t}\Phi \left( W_{<t>}[(0)v]\right) 
\end{eqnarray*}
\end{proof}

\begin{theorem}
\label{ThEigen}For every integer $t,$ $0\leq t\leq n,$ the $t$-th derivative 
$W_{<t>}$ belongs to the eigenspace of $1$ of $K^{[n-t]}.$
\end{theorem}

\begin{proof}
The statement of the theorem is true for $W_{<0>}=W$ from Lemma \ref{Lsim}.
Thus $WK^{[n]}=W.$ Since $W=(W[0]W[1])$ (recall that this is the
concatenation of $W[0]$\ and $W[1])$ and%
\begin{equation*}
K^{[n]}=K\otimes K^{[n-1]}=\frac{1}{\sqrt{2}}%
\begin{bmatrix}
K^{[n-1]} & K^{[n-1]} \\ 
K^{[n-1]} & -K^{[n-1]}%
\end{bmatrix}%
\end{equation*}%
we have 
\begin{eqnarray*}
\frac{1}{\sqrt{2}}(W[0]+W[1])K^{[n-1]} &=&W[0] \\
\frac{1}{\sqrt{2}}(W[0]-W[1])K^{[n-1]} &=&W[1].
\end{eqnarray*}%
By adding and then subtracting these equations, we obtain 
\begin{eqnarray*}
W[0]K^{[n-1]} &=&\frac{1}{\sqrt{2}}\left( W[0]+W[1]\right)  \\
W[1]K^{[n-1]} &=&\frac{1}{\sqrt{2}}\left( W[0]-W[1]\right) .
\end{eqnarray*}%
Now it is easy to check that 
\begin{eqnarray*}
\left( W[0]+\left( \sqrt{2}-1\right) W[1]\right) K^{[n-1]} &=&\frac{1}{\sqrt{%
2}}\left( W[0]+W[1]\right) +\frac{\sqrt{2}-1}{\sqrt{2}}\left(
W[0]-W[1]\right)  \\
&=&W[0]+\frac{2-\sqrt{2}}{\sqrt{2}}W[1] \\
&=&W[0]+\left( \sqrt{2}-1\right) W[1].
\end{eqnarray*}%
Hence, $W_{<1>}=$ $W_{<0>}[0]+\rho W_{<0>}[1]=W[0]+\left( \sqrt{2}-1\right)
W[1]$ is in the eigenspace of $1$ of $K^{[n-1]}.$ \newline
From this, we proof similarly that $W_{<2>}=$ $W_{<1>}[0]+\rho W_{<1>}[1]$
is in the eigenspace of $1$ of $K^{[n-2]}.$ \newline
Using Lemma \ref{nextD}, we complete the proof by induction on $t.$
\end{proof}

\section{A Balance Principle for Self-Dual Codes}

\begin{theorem}
Let $t$ be a selected coordinate position of $C,$ let $\delta \in F=\{0,1\},$
and let $A_{k,\delta }$ be the cardinality of the set 
\begin{equation*}
\left\{ v\in C~|~wt(v)=k\text{ and }v[t]=\delta \right\} .
\end{equation*}%
Then $\sum_{k=0}^{n}A_{k,1}\rho ^{k-1}=\sum_{k=0}^{n}A_{k,0}\rho ^{k+1}$ and 
$\sum_{k=0}^{n}A_{k,1}\rho ^{k-1}=\frac{1+\rho }{4}W_{<n>}.$ Particularly, $%
\sum_{k=0}^{n}A_{k,1}\rho ^{k-1}$ does not depend on the selection of $t.$
\end{theorem}

\begin{proof}
Let $t=n.$ From Theorem 6, $W_{<n-1>}=(W_{<n-1>}(0)W_{<n-1>}(1))$ is in the
eigenspace of $1$ of $K.$ which is generated by vector $(1,\rho ).$
Therefore, $W_{<n-1>}(1)=\rho W_{<n-1>}(0).$ Since 
\begin{eqnarray*}
W_{<n-1>}(1) &=&\sum_{u\in F^{n-1}}\rho
^{wt(u)}W[u(1)]=\sum_{k=0}^{n}A_{k,1}\rho ^{k-1} \\
W_{<n-1>}(0) &=&\sum_{u\in F^{n-1}}\rho
^{wt(u)}W[u(0)]=\sum_{k=0}^{n}A_{k,0}\rho ^{k}
\end{eqnarray*}%
we obtain the first equality of the Theorem. From Lemma 5, $%
W_{<n>}=W_{<n-1>}(0)+\rho W_{<n-1>}(1).$ Thus $W_{<n>}=W_{<n-1>}(0)+\rho
^{2}W_{<n-1>}(0).$ Using Lemma 4, we obtain $\sum_{k=0}^{n}A_{k}\rho
^{k}=(1+\rho ^{2})W_{<n-1>}(0)$ and $W_{<n-1>}(0)=\frac{1}{1+\rho ^{2}}%
\sum_{k=0}^{n}A_{k}\rho ^{k}.$ Then $W_{<n-1>}(1)=\frac{\rho }{1+\rho ^{2}}%
\sum_{k=0}^{n}A_{k}\rho ^{k}.$ Since $\frac{\rho }{1+\rho ^{2}}=\frac{1+\rho 
}{4},$ we obtain the second equality of the Theorem. \newline
If $t\neq n,$ we apply a permutation to the coordinates of $C$ which sends $%
t $ to $n.$ Since equivalent codes have the same weight distribution, the
statement of the Theorem holds for $t\neq n.$
\end{proof}

\section{Examples}

\subsection{Binary Self-Dual Codes of Length 8}

There are two binary self-dual codes of length 8 \cite{Pless1972}. They have
weight enumerators \newline
$W_{a}=(1,0,0,0,14,0,0,0,1)$ and $W_{b}=(1,0,4,0,6,0,4,0,1).$ Here the entry 
$A_{k}$ in position $1+k$ is equal to the number of code vectors of weight $%
k,$ $k=0,1,\ldots ,8.$

Each of these weight enumerator satisfies the equation $XM=2^{4}X$ where $M$
is the transpose of the Krawtchouk matrix $K_{8}$ \cite{Sim}. The entry of $%
K_{n}$ in row $i$ and column $j$ is 
\begin{equation*}
\sum_{m=0}^{i}(-1)^{m}\binom{j}{m}\binom{n-j}{i-m}.
\end{equation*}

Using computer algebra system Magma \cite{Mag}, we determine all solutions
of this equation with nonnegative integer entries. They are:\newline
$(1,0,0,0,14,0,0,0,1),(1,0,1,0,12,0,1,0,1),\newline
(1,0,20,10,0,20,1),(1,0,30,8,0,30,1),\newline
(1,0,4,0,6,0,40,1),(1,0,5,0,4,0,5,0,1),\newline
(1,0,6,0,20,6,0,1),(1,0,7,0,0,0,7,0,1).$

Let \ $A_{2,0}=y$ for a self-dual binary code realizing some of these
solutions. Then we can express the possible nonzero values of $A_{k,0}$ and $%
A_{k,1}$ with $y.$ We have

\begin{tabular}{|c|c|c|c|c|c|}
\hline
$k$ & $0$ & $2$ & $4$ & $6$ & $8$ \\ \hline
$A_{k,0}$ & $1$ & $y$ & $\frac{1}{2}A_{4}$ & $A_{2}-y$ & $0$ \\ \hline
$A_{k,1}$ & $0$ & $A_{2}-y$ & $\frac{1}{2}A_{4}$ & $y$ & $1$ \\ \hline
\end{tabular}

For each of the eight possible weight enumerators we solve for $y$ the
balance equation from Theorem 7. The found values of $y$ are $0,$ $\frac{3}{4%
},$ $\frac{3}{2},$ $\frac{9}{4},$ $3,$ $\frac{15}{4},$ $\frac{9}{2},$ $\frac{%
21}{4}.$ Since $y$ must be a nonnegative integer, this result eliminates all
possible weight enumerators except for $W_{a}$ and $W_{b}.$

\subsection{The Golay Code of Length 24}

This well known code is the only binary self-dual code of length 24 with
minimum weight 8 \cite{Pless1968}. The nonzero entries of the weight
distribution of this code are $A_{0}=1,$ $A_{8}=759,$ $A_{12}=2576,$ $%
A_{16}=759,$ $A_{24}=1.$ The supports of all code vectors of weight 8, 12,
and 16 form 5-designs \cite{AssMatt}. Hence we can compute all values $%
A_{k,0}$ and $A_{k,1.}$The result is given in the table\newline
\begin{tabular}{|c|c|c|c|c|c|}
\hline
$k$ & $0$ & $8$ & $12$ & $16$ & $24$ \\ \hline
$A_{k,0}$ & $1$ & $506$ & $\frac{2576}{2}$ & $759-506$ & $0$ \\ \hline
$A_{k,1}$ & $0$ & $759-506$ & $\frac{2576}{2}$ & $506$ & $1$ \\ \hline
\end{tabular}%
\newline
These values satisfy the balance equation of Theorem 7.

Using the three 5-design supported by the code wards of weight 8, 12, and
16, we compute the derivative $W_{<19>}$ and check that it satisfies the
equation $W_{<19>}K^{[5]}=$ $W_{<19>}.$ Its entries are given below:

-1167936*p + 483776, 1202240*p - 497984,

1202240*p - 497984, -1180608*p + 489024,

1202240*p - 497984, -1180608*p + 489024,

-1180608*p + 489024, 1023936*p - 424128,

1202240*p - 497984, -1180608*p + 489024,

-1180608*p + 489024, 1023936*p - 424128,

-1180608*p + 489024, 1023936*p - 424128,

1023936*p - 424128, -750144*p + 310720,

1202240*p - 497984, -1180608*p + 489024,

-1180608*p + 489024, 1023936*p - 424128,

-1180608*p + 489024, 1023936*p - 424128,

1023936*p - 424128, -750144*p + 310720,

-1180608*p + 489024, 1023936*p - 424128,

1023936*p - 424128, -750144*p + 310720,

1023936*p - 424128, -750144*p + 310720,

-750144*p + 310720, 7081024*p - 2933056

\subsection{The Extended Quadratic Residue Code of Length 48}

This is the only [48,24,12] binary self-dual code \cite{Houg}. The supports
of all code vectors of weight 12, 16, 20, 24, 28, 32, and 36 form 5-designs 
\cite{AssMatt}. This allows us to compute the derivative $W_{<43>}$ and to
check that $W_{<43>}K^{[5]}=$ $W_{<43>}.$ The entries of $W_{<43>}$ are

-1801295637184512*p + 746121082765312,

1799773991403520*p - 745490796445696,

1799773991403520*p - 745490796445696,

-1787591301070848*p + 740444560883712,

1799773991403520*p - 745490796445696,

-1787591301070848*p + 740444560883712,

-1787591301070848*p + 740444560883712,

1749670572392448*p - 724737280770048,

1799773991403520*p - 745490796445696,

-1787591301070848*p + 740444560883712,

-1787591301070848*p + 740444560883712,

1749670572392448*p - 724737280770048,

-1787591301070848*p + 740444560883712,

1749670572392448*p - 724737280770048,

1749670572392448*p - 724737280770048,

-1666630947373056*p + 690341141872640,

1799773991403520*p - 745490796445696,

-1787591301070848*p + 740444560883712,

-1787591301070848*p + 740444560883712,

1749670572392448*p - 724737280770048,

-1787591301070848*p + 740444560883712,

1749670572392448*p - 724737280770048,

1749670572392448*p - 724737280770048,

-1666630947373056*p + 690341141872640,

-1787591301070848*p + 740444560883712,

1749670572392448*p - 724737280770048,

1749670572392448*p - 724737280770048,

-1666630947373056*p + 690341141872640,

1749670572392448*p - 724737280770048,

-1666630947373056*p + 690341141872640,

-1666630947373056*p + 690341141872640,

11704454583615488*p - 4848143828713472

\subsection{The Putative [72,36,16] Binary Self-Dual Code}

The existence of this code is a long standing open question \cite{Sloane}.
Its weight enumerator is known:

[\TEXTsymbol{<}0,1\TEXTsymbol{>},\TEXTsymbol{<}16,249849\TEXTsymbol{>},%
\TEXTsymbol{<}20,18106704\TEXTsymbol{>},\TEXTsymbol{<}24,462962955%
\TEXTsymbol{>},\TEXTsymbol{<}28,4397342400\TEXTsymbol{>},\TEXTsymbol{<}%
32,16602715899\TEXTsymbol{>},

\TEXTsymbol{<}36,25756721120\TEXTsymbol{>},\TEXTsymbol{<}40,16602715899%
\TEXTsymbol{>},\TEXTsymbol{<}44,4397342400\TEXTsymbol{>},\TEXTsymbol{<}%
48,462962955\TEXTsymbol{>},\TEXTsymbol{<}52,18106704\TEXTsymbol{>},

\TEXTsymbol{<}56,249849\TEXTsymbol{>},\TEXTsymbol{<}72,1\TEXTsymbol{>}].

The supports all code vectors of weight 16, 20, 24, 28, 32, 36, 40, 44, 48,
52, and 56 form 5-designs \cite{AssMatt}. Using this designs, we compute the
derivative $W_{<67>}$ and check that $W_{<67>}K^{[5]}=$ $W_{<67>}.$ The
entries of $W_{<67>}$ are

-2766052144192871265730560*p + 1145736312355867181711360,

2764814255733869297795072*p - 1145223562167443435552768,

2764814255733869297795072*p - 1145223562167443435552768,

-2759776102730699957600256*p + 1143136690864219955920896,

2764814255733869297795072*p - 1145223562167443435552768,

-2759776102730699957600256*p + 1143136690864219955920896,

-2759776102730699957600256*p + 1143136690864219955920896,

2745563299524396139413504*p - 1137249555016829104029696,

2764814255733869297795072*p - 1145223562167443435552768,

-2759776102730699957600256*p + 1143136690864219955920896,

-2759776102730699957600256*p + 1143136690864219955920896,

2745563299524396139413504*p - 1137249555016829104029696,

-2759776102730699957600256*p + 1143136690864219955920896,

2745563299524396139413504*p - 1137249555016829104029696,

2745563299524396139413504*p - 1137249555016829104029696,

-2713300183161139309314048*p + 1123885734654746797539328,

2764814255733869297795072*p - 1145223562167443435552768,

-2759776102730699957600256*p + 1143136690864219955920896,

-2759776102730699957600256*p + 1143136690864219955920896,

2745563299524396139413504*p - 1137249555016829104029696,

-2759776102730699957600256*p + 1143136690864219955920896,

2745563299524396139413504*p - 1137249555016829104029696,

2745563299524396139413504*p - 1137249555016829104029696,

-2713300183161139309314048*p + 1123885734654746797539328,

-2759776102730699957600256*p + 1143136690864219955920896,

2745563299524396139413504*p - 1137249555016829104029696,

2745563299524396139413504*p - 1137249555016829104029696,

-2713300183161139309314048*p + 1123885734654746797539328,

2745563299524396139413504*p - 1137249555016829104029696,

-2713300183161139309314048*p + 1123885734654746797539328,

-2713300183161139309314048*p + 1123885734654746797539328,

18297763639587213294960640*p - 7579181860614308846632960

The existence of this code remains an open question. If a code exists, there
must be a chain of derivatives $\ W_{<67>},\cdots ,W_{<0>}$ with $W_{<0>}$
having only $0$ and $1$ entries. Lemma \ref{nextD}, Lemma \ref{2half}, and
Theorem \ref{ThEigen} may be usful tools in a search for such chain.

\begin{acknowledgement}
This copy of computer algebra system Magma \cite{Mag} used in this work has
been made available through a generous initiative of the Simons Foundation
covering U.S. colleges, universities, nonprofit research entities, and their
students, faculty, and staff.
\end{acknowledgement}

\end{document}